\documentclass{eptcs}
\providecommand{\event}{ICE 2012} 
\usepackage{breakurl}             

\usepackage{graphicx,color}
\usepackage{amsmath,amsthm,amscd}
\usepackage{latexsym,amssymb}
\usepackage{enumitem}

\theoremstyle{plain}
\newtheorem{theorem}{Theorem}[section]
\newtheorem{proposition}[theorem]{Proposition}
\newtheorem{lemma}[theorem]{Lemma}
\newtheorem{hypothesis}[theorem]{Hypothesis}

\theoremstyle{definition}
\newtheorem{definition}[theorem]{Definition}

\newtheorem{example}[theorem]{Example}
\newtheorem{note}[theorem]{Remark}

\newcommand{\defn}[1]{Definition~\ref{defn:#1}}

\newcommand{\tab}[1]{Table~\ref{tab:#1}}
\newcommand{\eq}[1]{(\ref{eq:#1})}

\newcommand{\ex}[1]{Example~\ref{ex:#1}}
\newcommand{\secn}[1]{Section~\ref{sec:#1}}
\newcommand{\rem}[1]{Remark~\ref{rem:#1}}
\newcommand{\lem}[1]{Lemma~\ref{lem:#1}}

\newcommand{\thm}[1]{Theorem~\ref{thm:#1}}

\font\tensym=msbm10
\font\sevensym=msbm7
\font\fivesym=msbm5

\newfam\symfam
\textfont\symfam=\tensym
\scriptfont\symfam=\sevensym
\scriptscriptfont\symfam=\fivesym
\def\sym{\fam\symfam\tensym}

\newcommand{\sB}{\ensuremath{\mbox{\sym B}}}
\newcommand{\bB}{\ensuremath{\mathbf{B}}}
\newcommand{\cB}{\mbox{\ensuremath{\cal B}}}

\newcommand{\cP}{\mbox{\ensuremath{\cal P}}}

\newcommand{\cR}{\mbox{\ensuremath{\cal R}}}

\newcommand{\mdash}[1][]{---#1}
\newcommand{\ndash}[1][]{--#1}
\newcommand{\etc}{etc.\ }
\newcommand{\ie}{i.e.\ }
\newcommand{\cf}{cf.\ }

\newcommand{\ndf}{\ensuremath{NDF'}}

\newcommand{\id}{\ensuremath{\mathbf{Id}}}

\newcommand{\goesto}[1][]{\stackrel{#1}{\rightarrow}^{}}
\newcommand{\longgoesto}[1][]{\stackrel{#1}{\longrightarrow}}

\newcommand{\bydef}[1]{\ensuremath{\stackrel{def}{#1}}}

\newcommand{\glue}[1][]{\ensuremath{G^{#1}}}

\newcommand{\comp}[1][]{\ensuremath{{\cal C}^{#1}}}

\newcommand{\actions}{\ensuremath{\mathbf A}}
\newcommand{\powerset}{\ensuremath{\cP_{\omega}}}
\newcommand{\zerobeh}{\ensuremath{\mathbf{0}}}

\newcommand{\issimulatedby}{\ensuremath{\sqsubseteq}}
\newcommand{\lessint}{\ensuremath{\preccurlyeq}}
\newcommand{\sync}[2]{\ensuremath{
    \renewcommand{\arraystretch}{0.3}
    \begin{array}{@{\ }c@{\ }}
      {\scriptstyle #1}\\
      {\bowtie}\\
      {\scriptstyle #2}
    \end{array}
  }
}
\newcommand{\maxwit}[1]{\ensuremath{\stackrel{#1}{\cup}}}
\newcommand{\meet}{\ensuremath{\otimes}}

\newcommand{\bigmeet}{\ensuremath{\bigotimes}}

\newcommand{\arity}[2]{\ensuremath{{#2}^{(#1)}}}
\newcommand{\subsequences}[2][\ ]{\ensuremath{{#1\wr}_{#2}}}
\newcommand{\factor}[1]{\ensuremath{{#1}/\!\!\simeq}}

\newcommand{\derrule}[3][1]{%
  \ensuremath{%
    \begin{array}{*{#1}{@{\hspace{3mm}}c@{\hspace{3mm}}}}
      #2\\
      \hline 
      \multicolumn{#1}{c}{#3}
    \end{array}%
  }%
}

\newlength{\lnframe}
\newlength{\lnframep}

\definecolor{mycgreen}{rgb}{0.13,0.54,0.06}

\title{Towards a Theory of Glue}
\author{Simon Bliudze
\institute{\'Ecole Polytechnique F\'ed\'erale de Lausanne\\
Rigorous System Design Laboratory\\
INJ Building, Station 14, 1015 Lausanne, Switzerland}
\email{simon.bliudze@epfl.ch}
}
\def\titlerunning{Towards a Theory of Glue}
\def\authorrunning{S. Bliudze}

\begin{document}
\maketitle

\begin{abstract}
  We propose and study the notions of {\em behaviour type} and {\em
    composition operator} making a first step towards the definition
  of a formal framework for studying behaviour composition in a
  setting sufficiently general to provide insight into how the
  component-based systems should be modelled and compared.  We
  illustrate the proposed notions on classical examples (Traces,
  Labelled Transition Systems and Coalgebras).  Finally, the
  definition of {\em memoryless glue operators}, takes us one step
  closer to a formal understanding of the separation of concerns
  principle stipulating that computational aspects of a system should
  be localised within its atomic components, whereas coordination
  layer responsible for managing concurrency should be realised by
  memoryless glue operators.
\end{abstract}


\section{Introduction} \label{sec:intro}

Component-based design is central in system engineering.  Complex
systems are built by assembling components.  Large components are
obtained by ``gluing'' together simpler ones.  ``Gluing'' can be
considered as an operation on sets of components.

Component-based techniques have seen significant development,
especially through the use of object technologies supported by
languages such as C++, Java, and standards such as UML and CORBA.
There exist various component frameworks encompassing a large variety
of mechanisms for composing components.  They focus rather on the way
components interact than on their internal behaviour.  We lack
adequate notions of expressiveness to compare the merits and
weaknesses of these frameworks.  For a meaningful and systematic
comparison of component frameworks to be possible, one needs a
sufficiently abstract formalisation of the notions of both behaviour
and glue.

These notions should capture the properties essential for reasoning
about system composition.  Intuitively, one can think of coordination
as imposing constraints on the joint behaviour of the components of
the system \cite{BliSif11-constraints-sc, DeconstructingReo,
  Wegner96}.  Imposing coordination constraints means ``reducing'' the
joint behaviour of the involved components.  Beyond the simple
definition of a component model, this requires the following two
questions to be answered:
\begin{itemize}
\item {\em What is the behaviour of several ``parallel'' components
  without any coordination constraints?}
\item {\em How does one compare two behaviours?}
\end{itemize}
We argue that these aspects, as well as the notions of common
behaviour of two components and minimal behaviour possible in a given
framework, cannot be dissociated from the notion of behaviour as a
whole.

There are several goals, for which the work in this paper is a
starting point.

\begin{itemize}
\item There is a long-standing debate with a plethora of points of
  view about {\em separation of concerns} \cite{montanari06,
    clements95-fromsubroutines, coulson08, mckinley04-adaptive-sw,
    rashid02-requirements} and in what sense it should be respected.
  In \cite{BliSif07-acp-emsoft, BliSif10-causal-fmsd}, we have
  advocated for an approach respecting a separation of concerns
  principle, whereby all ``computation'' performed by the system is
  localised within its constituent atomic components, whereas the
  coordination layer responsible for managing the parallelism consists
  of {\em memoryless} glue operators.  Separation of concerns applied
  in the context of the BIP framework resulted in a very powerful
  deadlock detection tool D-Finder \cite{D-Finder2}.  Furthermore, it
  has allowed us to reduce a hard (sometimes undecidable) problem of
  synthesis of reactive systems \cite{vardi09-libraries,
    pnueli90-synthesis-hard} to a less ambitious, but more tractable
  design methodology \cite{BliSif11-constraints-sc}.  This separation
  of concerns principle is rather fragile: in \cite{diGiusto11}, it is
  shown that a slight extension of the BIP glue renders it Turing
  complete, which makes it possible to construct non-trivial systems
  without a single atomic component.  We speculate that, in the
  setting proposed in the present paper, operators such as prefixing
  and choice are not glue operators.  This hypothesis, should it be
  verified, would give formal grounds to our view of the separation of
  concern principle.

\item Another goal is to define a generic setting for an
  expressiveness study. In \cite{BliSif08-express-concur}, we have
  proposed a first notion of glue expressiveness for component-based
  systems, generalising the concepts presented in \cite{Sifakis} that
  guided the design of the BIP (Behaviour, Interaction, Priority)
  framework \cite{bip06}.  However, this approach lacks abstraction,
  since it strongly depends on the choice of the formalism for
  modelling both behaviour (Labelled Transition Systems; LTS) and glue
  (Structural Operational Semantics; SOS \cite{plotkin81structural}),
  which makes the expressiveness comparison questionable and dependent
  on ad-hoc manipulation to make different frameworks comparable.  For
  instance, in a more recent paper \cite{BliSif11-constraints-sc}, we
  have proposed a slightly modified formal model of component
  behaviour in BIP.  This entailed a corresponding modification to the
  definition of glue operators, which resulted in a framework only
  partially comparable to that in \cite{BliSif08-express-concur}.  A
  general definition of behaviour and glue, as in the present paper,
  is necessary to solve this problem.

\item Finally, a common practice in system engineering consists in
  applying existing solutions (architectures) to given components to
  ensure behavioural properties (mutual exclusion, scheduling,
  communication protocols, etc.). These architectures can be modelled
  as composition operators.  Their simultaneous application should
  then ensure the application of the coordination constraints imposed
  by both operators.  Hence, it is important to understand when and
  how these architectures can be combined.  This brings forward
  another question: {\em How does one model the simultaneous
    application of several composition operators and what are the
    conditions ensuring the non-interference among them?}
\end{itemize}

This paper is a modest first step towards the above goals.  It focuses
primarily on the notions of {\em behaviour types} and {\em composition
  operators} defined in such a way as to allow reasoning about their
essential properties in a setting sufficiently general to provide
insight into how component-based systems should be modelled and
compared.

To allow answering the questions that we have emphasised above, a
behaviour type $\cB$ must explicitly comprise certain elements beyond
the minimal component model such as LTS.  As illustrated, for
instance, by CCS and SCCS \cite{milner83-calculi}, different parallel
composition operators can be defined on the same objects.  Hence a
parallel composition operator $\parallel$ must be defined as part of a
behaviour type.  In order to avoid confusion with any of the existing
parallel composition operators, we use the term {\em maximal
  interaction operator}, since, intuitively, the operator defines the
maximal set of interactions between two components in absence of
coordination constraints.

In order to address the question of comparing the behaviour, we
require that two preorders be defined: a {\em simulation preorder}
$\issimulatedby$ and a {\em semantic preorder} $\lessint$.
Intuitively, simulation preorder relates two components if one of them
performs only actions that can also be performed by the other, thus
generalising the classical simulation relation and allowing a
formalisation of the notion that composing component behaviour amounts
to imposing coordination constraints.  The role of the semantic
preorder is to relate components that behave in a similar manner, in
particular, the equivalence $\simeq$ induced by this preorder is a
congruence with respect to composition operators, thus generalising
such classical notions as ready simulation equivalence and
bisimilarity.

Finally, we need a meet operator $\meet$, rendering $(\factor{\cB},
\lessint, \meet)$ a meet-semilattice, to talk about simultaneous
application of two composition operators (\secn{combining}).  Let
$f_1$ and $f_2$ be two composition operators, which we would like to
apply "simultaneously" to a behaviour $B$.  Viewing composition
operators as constraints on $B$, simultaneous application of the two
constraints corresponds intuitively to applying their conjunction.
However, this idea does not fit into the functional view of operators,
since neither $f_1(f_2(B))$ nor $f_2(f_1(B))$ need correspond to the
conjunction of the two constraints.  Moreover, neither of these two
behaviours need be defined, in particular, since the arity constraints
of the operators are not respected in general.  A more direct approach
consists in considering the {\em maximal common behaviour} of $f_1(B)$
and $f_2(B)$, which is precisely $f_1(B) \meet f_2(B)$. 

The paper is structured as follows.  \secn{bt} introduces the notion
of behaviour type and provides three examples: traces, LTS and
coalgebras of a particular type.  \secn{composition} presents the
notion of composition operators and some of their properties.  In
\secn{discussion}, we discuss some related and future work, then we
conclude in \secn{conclusion}.


\section{Behaviour types}
\label{sec:bt}

\begin{definition}[Behaviour type]
  \label{defn:behaviour}
  A \emph{behaviour type} over $\cB$ is a tuple $(\cB, \parallel,
  \issimulatedby, \lessint, \meet, \zerobeh)$ consisting of the following
  data:
  \begin{itemize}
  \item A monoid $(\cB, \parallel, \zerobeh)$, where $\parallel\, : \cB^2
    \rightarrow \cB$ is a totally defined associative binary operator with
    $\zerobeh \in \cB$ the neutral element;
  \item A preorder $\issimulatedby\ \subseteq \cB \times \cB$, such that
    \begin{enumerate}
    \item \label{parallel-simulation} it is preserved by $\parallel$, \ie
      for any $B_1,B_2,B_3 \in \cB$, $B_2 \issimulatedby B_3$ implies $B_1
      \parallel B_2 \issimulatedby B_1 \parallel B_3$,
    \item \label{parallel-bottom} for any $B \in \cB$, holds $\zerobeh
      \issimulatedby B$;
    \end{enumerate}
  \item A preorder $\lessint\ \subseteq \cB \times \cB$ and a meet
    operator $\meet: \cB \times \cB \rightarrow \cB$, such that
    \begin{enumerate}
    \setcounter{enumi}{2}
    \item \label{parallel-semantic} $\lessint$ is preserved by $\parallel$,
      \ie for any $B_1,B_2,B_3 \in \cB$, $B_2 \lessint B_3$ implies $B_1
      \parallel B_2 \lessint B_1 \parallel B_3$,
    \item \label{meet} $(\factor{\cB}, \lessint, \meet)$, with $\simeq
      \bydef{=} \lessint \cap \lessint^{-1}$, is a meet-semilattice.
    \end{enumerate}
  \end{itemize}

  Elements of $\cB$ are \emph{behaviours}.  $\parallel$ is the
  \emph{maximal interaction operator}. Intuitively, $B_1 \parallel B_2$ is
  the parallel composition of $B_1$ and $B_2$ in absence of any
  coordination constraints.  $\issimulatedby$ is the {\em simulation
    preorder}.  It has the same meaning as the classical simulation
  preorder, \ie $B_1 \issimulatedby B_2$ means that ``$B_2$ can perform at
  least the same executions as $B_1$''.  $\lessint$ is the {\em semantic
    preorder}.  Intuitively, $B_1 \lessint B_2$ means that $B_2$ can act as
  $B_1$ in the same coordination context.  $B_1 \meet B_2$ is the greatest
  behaviour realisable by both $B_1$ and $B_2$.
\end{definition}

Often behaviour definition explicitly involves an interface, consisting at
least of the set of actions a component can perform.  In this context, it
is usually assumed that a universal set of actions is given.  Furthermore,
it is convenient also to assume that this set is equipped with some lattice
structure.  A straightforward example consists in considering finite sets
rather than individual actions with the lattice structure induced by set
union and intersection.

\subsection{Example: Traces}
\label{sec:traces}

Let $\actions$ be a universal set of actions and $Traces$ be the set
of pairs $B = (A, T)$, where $A \subseteq \actions$ is a set of
actions and $T \subseteq A^*$ is a prefix-closed set of traces.  In
particular, $\varepsilon \in T$, where $\varepsilon$ is the empty
word.  For $B_1 = (T_1, A_1)$ and $B_2 = (T_2, A_2)$, we define the
maximal interaction operator $\parallel$, as the interleaving of
actions of the two behaviours, by putting $B_1 \parallel B_2 \bydef{=}
(A_1 \cup A_2, T)$, where
\begin{eqnarray}
  \label{eq:traces:parallel}
    T & \bydef{=} &
    \Big\{
    w = (w_i)_{i=1}^n \in A^*
    \,\Big|\,
    \exists I \subseteq [1,n]: (w_i)_{i \in I} \in T_1 
    \land (w_i)_{i \not\in I} \in T_2
    \Big\}\,.    
\end{eqnarray}
This makes $(Traces, \parallel, \zerobeh)$, with $\zerobeh = (\emptyset,
\{\varepsilon\})$, a commutative monoid.
 
Equation \eq{traces:parallel} defines the maximal interaction operator
through interleaving of traces.  The corresponding semantic preorder is
defined through the notion of sub-sequence, which we denote $\propto$ and
define by putting, for $v, w \in A^*$,
\begin{eqnarray*}
  v \propto w 
  & \bydef{\Longleftrightarrow} &
  \exists I \subseteq [1,|w|]: v = (w_i)_{i\in I}\,,
\end{eqnarray*}
where $|w|$ is the length of $w$.  Furthermore, for an alphabet $A$ and a
language $T$, we define $\subsequences[T]{A} \bydef{=} \{v \in A^*\,|\,
\exists w \in T: v \propto w\}$ (notice that, if $T \subseteq A^*$, we have
$T \subseteq \subsequences[T]{A}$).  Finally, we define the semantic
preorder and the operator $\meet$ as follows:
\begin{eqnarray}
  \label{eq:traces:lessint}
  B_1 \lessint B_2 & \bydef{\Longleftrightarrow} &
  A_1 \subseteq A_2 \ \land \ T_1 \subseteq \subsequences[T_2]{A_1}\,,\\
  \label{eq:traces:meet}
  B_1 \meet B_2 & \bydef{=} & (A_1\cap A_2, \subsequences[T_1]{A_1 \cap A_2} 
  \cap \subsequences[T_2]{A_1 \cap A_2})\,,
\end{eqnarray}
the preorder $\issimulatedby$ coincides with $\lessint$.

Clearly, for two alphabets $A \subseteq B$ and a set of traces $T$,
holds the equality $\subsequences[{(\subsequences[T]{B})}]{A} =
\subsequences[T]{A}$, which implies immediately that the operator
$\meet$ defined by \eq{traces:meet} is, indeed, the meet operator with
respect to the preorder $\lessint$.  Hence, $(\factor{Traces},
\lessint, \meet)$ is a meet-semilattice.  To prove that $(Traces,
\parallel, \issimulatedby,$ $\lessint, \meet, \zerobeh)$ is a
behaviour type we only have to show that condition
\ref{parallel-simulation} (and condition \ref{parallel-semantic}) of
\defn{behaviour} holds.

\begin{proposition}
  Maximal interaction operator $\parallel$ defined by \eq{traces:parallel}
  preserves the semantic preorder $\lessint$ defined by
  \eq{traces:lessint}.
\end{proposition}
\begin{proof}
  For $i=1,2,3$, let $B_i = (A_i, T_i) \in Traces$ be such that $B_2
  \lessint B_3$.  By definition \eq{traces:lessint} of $\lessint$, we
  then have $A_2 \subseteq A_3$ and $T_2 \subseteq
  \subsequences[T_3]{A_2}$.  Consequently, $A_1 \cup A_2 \subseteq A_1
  \cup A_3$.

  Let $B_1 \parallel B_2 = (A_1 \cup A_2, T_{12})$ and $B_1 \parallel
  B_3 = (A_1 \cup A_3, T_{13})$, and consider $w = (w_i)_{i=1}^n \in
  T_{12}$. By \eq{traces:parallel}, there exists $I \subset [1,n]$
  such that $(w_i)_{i\in I} \in T_1$ and $(w_i)_{i \not\in I} \in
  T_2$.  Since $T_2 \subseteq \subsequences[T_3]{A_2}$, there exists
  $v \in T_3$ such that $(w_i)_{i \not\in I} \propto v$.  Denoting by
  $\tilde{v}$ the projection of $v$ on $A_1 \cup A_2$, we have
  $\tilde{v} \propto v$, which implies $\tilde{v} \in
  \subsequences[T_3]{A_1 \cup A_2}$.  It is clear that $(w_i)_{i
    \not\in I} \propto \tilde{v}$ and that one can interleave
  $\tilde{v}$ with $(w_i)_{i \in I}$ in such a manner that the
  positions of $w_i$, for all $i \not \in I$, be preserved.  Denoting
  the obtained trace by $u$, we obviously obtain $w \propto u \in
  \subsequences[T_1]{A_1 \cup A_2} \subsequences[T_3]{A_1 \cup A_2} =
  \subsequences[T_{13}]{A_1 \cup A_2}$.  Hence $T_{12} \subseteq
  \subsequences[T_{13}]{A_1 \cup A_2}$ and $B_1 \parallel B_2 \lessint
  B_1 \parallel B_3$.
\end{proof}

\subsection{Example: Labelled Transition Systems}
\label{sec:lts}

Let again $\actions$ be a universal set of actions and let $LTS$ be
the set of triples $B=(Q,A,\goesto)$, with $Q \neq \emptyset$ a finite
set of states, $A \subseteq \actions$ a set of actions, and
$\goesto\ \subseteq Q \times 2^A \times Q$ a set of transitions.  The
infix notation $q \goesto[a] q'$ is commonly used to denote a
transition $(q, a, q') \in \goesto$.  Note that, with this definition,
transitions are labelled with \emph{interactions}, \ie sets of
actions.  This approach is particularly well suited for characterising
behaviour of components that can communicate through several ports
during a single transition \cite{BliSif07-acp-emsoft}.  Such behaviour
naturally arises when composite components are assembled by parallel
composition of sub-components.  Thus, labels are combined by set
union, often denoted by juxtaposition.

As in the previous example, we take coinciding preorders
$\issimulatedby$ and $\lessint$.  For $B_1 = (Q_1, A_1, \goesto_1)$
and $B_2 = (Q_2, A_2, \goesto_2)$, we define the maximal interaction
operator $\parallel$ and the semantic preorder $\lessint$ as follows.
We put $B_1 \parallel B_2 \bydef{=} (Q_1 \times Q_2,\ A_1 \cup
A_2,\ \goesto)$, where $\goesto$ is the minimal transition relation
satisfying the following SOS rules:

\begin{equation}
  \label{eq:lts:parallel}
  \begin{array}{*{3}{c}}
    \derrule{q_1 \goesto[a]_1 q_1'}{q_1q_2 \longgoesto[a] q_1' q_2}\,, &
    \derrule{q_2 \goesto[b]_2 q_2'}{q_1q_2 \longgoesto[b] q_1 q_2'}\,, &
    \derrule[2]
            {q_1 \goesto[a]_1 q_1' & q_2 \goesto[b]_2 q_2'}
            {q_1q_2 \longgoesto[a \cup b] q_1' q_2'}\,.
  \end{array}
\end{equation}

With $\zerobeh = (\{1\}, \emptyset, \emptyset)$, it is straightforward to
conclude that $(LTS, \parallel, \zerobeh)$ is a commutative monoid.

For $B_1 = (Q_1, A_1, \goesto_1)$ and $B_2 = (Q_2, A_2, \goesto_2)$, such
that $A_1 \subseteq A_2$ consider the maximal {\em simulation} relation
$R \subseteq Q_1 \times Q_2$ such that
\begin{eqnarray}
  \label{eq:lts:less}
  q_1 R q_2 & \Longrightarrow &
  \forall q_1 \goesto[a]_1 q_1',\   
  \exists\, q_2' \in Q_2, b \subseteq A_2:
  \left(q_2 \goesto[b] q_2'\ \land\ a\subseteq b\ \land\ q_1' R q_2'\right)\,.
\end{eqnarray}
The semantic preorder $\lessint$, is defined by putting $B_1 \lessint B_2$
iff $R$ is total on $Q_1$.

Finally, we define the meet operator by putting, for $B_1 = (Q_1, A_1,
\goesto_1)$ and $B_2 = (Q_2, A_2, \goesto_2)$, $B_1 \meet B_2 \bydef{=}
(Q_1 \times Q_2, A_1 \cap A_2, \goesto)$, where $\goesto$ is the minimal
transition relation satisfying the rule
\begin{equation}
  \label{eq:lts:meet}
  \derrule[3] {q_1 \goesto[a]_1 q_1' & 
    q_2 \goesto[b]_2 q_2' & 
    a \cap b \neq \emptyset
  }{q_1q_2 \longgoesto[a \cap b] q_1'q_2'}\,.  
\end{equation}

Conditions \ref{parallel-simulation}\ndash\ref{parallel-semantic} of the
\defn{behaviour} clearly hold, and, to show that $(LTS, \parallel,
\issimulatedby, \lessint, \meet, \zerobeh)$ is a behaviour type, we only
have to prove the following proposition.

\begin{proposition}
  $(\factor{LTS}, \lessint, \meet)$ is a meet-semilattice.
\end{proposition}
\begin{proof}
To prove the proposition, we have to show that, for any $\tilde{B},
B_1,B_2 \in LTS$, holds $B_1 \meet B_2 \lessint B_1, B_2$ and $\tilde{B}
\lessint B_1, B_2$ implies $\tilde{B} \lessint B_1 \meet B_2$.

Let $\tilde{B} = (\tilde{Q}, \tilde{A}, \goesto)$, $B_i = (Q_i, A_i,
\goesto)$ (for $i=1,2$) and $B_1 \meet B_2 = (Q_1 \times Q_2, A_1 \cap A_2,
\goesto)$ as defined above (for clarity, we skip the indices on the
transition relations).

1. First of all $A_1 \cap A_2 \subseteq A_1, A_2$. By symmetry, it is
sufficient to show that $B_1 \meet B_2 \lessint B_1$.  Let $R \subseteq
(Q_1 \times Q_2) \times Q_1$ be the projection on the first component.

Consider a state $q = (q_1, q_2) \in Q_1\times Q_2$. By definition of
$\meet$, for any $q \goesto[c] q'$ in $B_1 \meet B_2$, there exist
transitions $q_1 \goesto[a] q_1'$ and $q_2 \goesto[b] q_2'$ such that $a
\cap b = c$ and $q' = (q_1', q_2')$.  Hence, we have $c \subseteq a$ and
$q'R q_1'$, satisfying the implication \eq{lts:less} is satisfied. Since
$R$ is total on $Q_1 \times Q_2$, we have $B_1 \meet B_2 \lessint B_1$.

2. Since $\tilde{B} \lessint B_1, B_2$, there exist total relations $R_i
\subseteq \tilde{Q} \times Q_i$ (for $i=1,2$) satisfying the implication
\eq{lts:less}.  We define a relation $R \subseteq \tilde{Q} \times (Q_1,
Q_2)$ by putting $\tilde{q} R (q_1,q_2)$ iff $\tilde{q} R_i q_i$, for
both $i=1,2$.  Since $\tilde{A} \subseteq A_1 \cap A_2$ and $R$ is
clearly total, we only have to show that \eq{lts:less} is satisfied.

Let $\tilde{q} \in \tilde{Q}$ and $(q_1,q_2) \in Q_1 \times Q_2$ be related
by $R$.  Consider a transition $\tilde{q} \goesto[c] \tilde{q}'$ in
$\tilde{B}$.  Since $\tilde{q} R_i q_i$, for $i=1,2$, there exist two
transitions $q_1 \goesto[a] q_1'$ and $q_2 \goesto[b] q_2'$ such that $c
\subseteq a, b$ and $\tilde{q}' R_i q_i'$. We then have $c \subseteq a \cap
b$ and $\tilde{q}' R (q_1', q_2')$, which satisfies \eq{lts:less} and
proves the proposition.
\end{proof}


\subsection{Example: Coalgebras}
\label{sec:coalgebras}

Coalgebras \cite{Rutten00} provide a general framework for unifying
various state-based behaviour models such as both deterministic and
non-deterministic automata, LTS, Mealy machines, etc.  The
presentation below is largely inspired by that in \cite{Silva-PhD}.

We recall, in \tab{operations}, the extensions from sets to mappings
of the following four operations: Cartesian product $X \times Y$,
disjoint union $X + Y$, exponentiation $X^A$ (set of mappings $A
\rightarrow X$) and powerset $\powerset(X)$ (set of finite subsets of
$X$).

\begin{table}
  \caption{Extensions of set operations to mappings ($f: X \rightarrow
    Y$, $f_1: X_1 \rightarrow Y_1$ and $f_2: X_2 \rightarrow Y_2$)}
  \label{tab:operations}
  \centering
  \begin{tabular}{|*{2}{@{\qquad}c@{\qquad}|}}
    \hline
    $
    \begin{gathered}
      f_1 \times f_2 : X_1 \times X_2 \rightarrow Y_1 \times Y_2\\
      (x_1, x_2) \mapsto \Big(f_1(x_1), f_2(x_2)\Big)
    \end{gathered}
    $
    &
    $
    \begin{gathered}
      f_1 + f_2 : X_1 + X_2 \rightarrow Y_1 + Y_2\\
      x \mapsto \kappa_i(f_i(x'))
    \end{gathered}
    $
    \\
    \hline
    $
    \begin{gathered}
      f^A : X^A \rightarrow Y^A\\
      g \mapsto f \circ g
    \end{gathered}
    $
    &
    $
    \begin{gathered}
      \powerset(f) : \powerset(X) \rightarrow \powerset(Y)\\
      S \mapsto \{f(x)\,|\, x\in S\}
    \end{gathered}
    $
    \\
    \hline
  \end{tabular}
\end{table}

Let $F$ be a functor on $\mathbf{Set}$.  An $F$-coalgebra is a pair
$(S, f: S \rightarrow F(S))$, where $S$ is the set of states.  The
mapping $f$ determines the transition structure of $(S,f)$, whereas
the functor $F$ is the {\em type} of the coalgebra.

\begin{example}
  Functors $M = (\bB \times \id)^\actions$, $D = \sB \times (1 +
  \id)^\actions$ and $N = \sB^2 \times \powerset(\id)^\actions$, where
  $\sB = \{0,1\}$ and $1 = \{*\}$, are respectively the types of the
  coalgebraic definitions for, respectively, input-enabled Mealy
  machines with the input domain $\actions$ and output domain $\bB$,
  deterministic and non-deterministic automata.  For instance, in a
  $D$-coalgebra $(S, f: S \rightarrow D(S))$, the components of the
  mapping $f = \omega \times \delta$, with $\omega: S \rightarrow \sB$
  and $\delta: S \rightarrow (1 + S)^A$, determine, for each state $s
  \in S$, whether it is a final state (mapping $\omega$) and the set
  $\Big\{\Big(s, a, \delta(s)(a)\Big) \,\Big|\, a \in A, \delta(s)(a)
  \in S\Big\}$ of transitions leaving $s$.  The case where there is no
  transition labelled $a$ leaving the state $s$ is reflected by
  $\delta(s)(a) = * \in 1$.  Notice that, in this example, we do not
  account for initial states of automata; determinism, here, means
  that only one transition is possible from any given state with a
  given action.
\end{example}

We define a class $\ndf$ of {\em non-deterministic functors}\footnote{
  $\ndf$ is a subclass of the class $NDF$ of non-deterministic functors as
  defined, for example, in \cite{Silva-PhD}.
} on the category $\mathbf{Set}$ of sets, defined by the following
grammar
\begin{equation}
  \label{eq:ndf}
  F ::= \id \,|\, \bB \,|\, F \times F \,|\, F^{\actions}
  \,|\, \powerset(F)
  \,,
\end{equation}
where $\actions$ is the universal set of actions, $\id$ is the identity
functor and $\bB \neq \emptyset$ is a join-semilattice with bottom.
Typical examples of semilattices used to define non-deterministic functors
are $\sB = \{0,1\}$ with $0 \lor 1 = 1$ and $0 \land 1 = 0$; and the
trivial lattice $1 = \{*\}$.  The functors in $NDF-$ are defined by
structural induction in \tab{ndf} (\cf also \tab{operations}).

\begin{table}
  \caption{Structural induction definition of $\ndf$ functors (for any
    set $X$ and any mapping $f: X \rightarrow Y$)}
  \label{tab:ndf}
  \centering
  \begin{tabular}{*{2}{|c}|l|}
    \hline
    $F(X)$ & $F(f)$ & Structure of $F$
    \\\hline\hline
    $X$ & $f$ & $F = \id$
    \\\hline
    $\bB$ & $id_{\bB}$ & $F = \bB$
    \\\hline
    $F_1(X) \times F_2(X)$ & $F_1(f) \times F_2(f)$ & $F = F_1 \times F_2$
    \\\hline
    $G(X)^\actions$ & $G(f)^\actions$ & $F = G^\actions$
    \\\hline
    $\powerset(G(X))$ & $\powerset(G(f))$ & $F = \powerset(G)$
    \\\hline
  \end{tabular}
\end{table}

Consider $F \in \ndf$ and let $(S_1, f_1)$ and $(S_2, f_2)$ be two
$F$-coalgebra.  For a relation $R \subseteq S_1 \times S_2$, we define
a relation $\leq_R^F\ \subseteq F(S_1) \times F(S_2)$ by structural
induction on $F$.  For $x \in F(S_1)$ and $y \in F(S_2)$, we put
\begin{eqnarray*}
  x \leq_R^F y & \bydef{\Longleftrightarrow} &
  \begin{cases}
    (x, y) \in R, & \text{if $F = \id$,}\\
    x \lor y = y, & \text{if $F = \bB$,}\\
    x_1 \leq_R^{F_1} y_1 \land x_2 \leq_R^{F_2} y_2, & 
    \text{if $F = F_1 \times F_2$, $x = (x_1,x_2)$ and $y = (y_1,y_2)$,}\\
    \forall a \in \actions,\ x(a) \leq_R^G y(a), & 
    \text{if $F = G^\actions$,}\\
    \forall x' \in x, \exists y' \in y: x' \leq_R^G y', &
    \text{if $F = \powerset(G)$.}
  \end{cases}
\end{eqnarray*}

We define the simulation preorder on $F$-coalgebra by putting $(S_1,
f_1) \issimulatedby (S_2, f_2)$, iff there exists a relation $R
\subseteq S_1 \times S_2$ total on $S_1$ and such that $\forall (s_1,
s_2) \in R,\ f_1(s_1) \leq_R^F f_2(s_2)$.

\begin{definition}[Coalgebra homomorphism]
  An {\em $F$-homomorphism} of two $F$-coalgebras $(S_1,f_1)$ and
  $(S_2,f_2)$ is a mapping $h: S_1 \rightarrow S_2$ preserving the
  transition structure, \ie such that the following diagram commutes
  ($f_2 \circ h = F(h) \circ f_1$)
  \[
    \begin{CD}
      S_1 @>h>> S_2\\
      @V{f_1}VV   @VV{f_2}V\\
      F(S_1) @>>{F(h)}> F(S_2)
    \end{CD}
  \]
  where $F(h)$ is the image of $h$ by the functor $F$.  In particular, for
  $F \in \ndf$, $F(h)$ is defined by the second column in \tab{ndf}.
\end{definition}

\begin{definition}[Bisimulation]
  \label{defn:bisimulation}
  Let $(S_1, f_1)$ and $(S_2, f_2)$ be two $F$-coalgebras.  A relation
  $R \subseteq S_1 \times S_2$ is a {\em bisimulation} iff there
  exists a {\em witness mapping} $g: R \rightarrow F(R)$, such that
  projections $\pi_i: R \rightarrow S_i$, for $i = 1,2$, are coalgebra
  homomorphisms, \ie the following diagram commutes
  \begin{equation}
    \label{eq:bisimulation}
    \begin{CD}
      S_1 @<{\pi_1}<< R @>{\pi_2}>> S_2\\
      @V{f_1}VV     @VgVV          @VV{f_2}V\\
      F(S_1) @<<{F(\pi_1)}< F(R) @>>{F(\pi_2)}> F(S_2)
    \end{CD}
  \end{equation}
\end{definition}

\begin{note}
  \label{rem:subbisimulation}
  In the context of \defn{bisimulation}, if $F = F_1 \times F_2$, $g =
  g_1 \times g_2$ and $f_i = f_i^1 \times f_i^2$ ($i=1,2$), the left
  diagram in \eq{subbisimulation} commutes, for $j=1,2$, and defines a
  bisimulation on $F_j$-coalgebras $(S_1, f_1^j)$ and $(S_2, f_2^j)$.
  Similarly, for $F = G^\actions$, the right diagram in
  \eq{subbisimulation} commutes, for all $a \in \actions$, and defines
  a bisimulation on $G$-coalgebras $(S_1, f_1(a))$ and $(S_2,
  f_2(a))$.
  \begin{align}
    \label{eq:subbisimulation}
    \begin{CD}
      S_1 @<{\pi_1}<< R @>{\pi_2}>> S_2\\
      @V{f_1^j}VV     @Vg_jVV          @VV{f_2^j}V\\
      F_j(S_1) @<<{F_j(\pi_1)}< F_j(R) @>>{F_j(\pi_2)}> F_j(S_2)
    \end{CD}
    &&
    \begin{CD}
      S_1 @<{\pi_1}<< R @>{\pi_2}>> S_2\\
      @V{f_1(a)}VV     @Vg(a)VV          @VV{f_2(a)}V\\
      G(S_1) @<<{G(\pi_1)}< G(R) @>>{G(\pi_2)}> G(S_2)
    \end{CD}
  \end{align} 
\end{note}

\begin{lemma}[\cite{Rutten00}]
  \label{lem:bisim:img}
  Consider coalgebras $(S, f)$, $(S_1,f_1)$ and $(S_2, f_2)$ with
  coalgebra homomorphisms $g: S \rightarrow S_1$ and $h: S \rightarrow
  S_2$.  The image $(g,h)(S) = \{(g(s),h(s))\,|\, s\in S\}$ is a
  bisimulation on $S_1$ and $S_2$.
\end{lemma}

\begin{theorem}[\cite{Rutten00}]
  \label{thm:bisim:union}
  The union $\bigcup_k R_k$ of a family $\{R_k\}_k$ of bisimulations
  on coalgebras $(S_1,f_1)$ and $(S_2,f_2)$ is again a bisimulation.
\end{theorem}

\begin{lemma}
  \label{lem:bisim:comp}
  Consider $F$-coalgebras $(S_1, f_1)$, $(S_2,f_2)$ and $(S_3, f_3)$
  with $F \in \ndf$.  The composition $R_1 \circ R_2$ of two
  bisimulations $R_1 \subseteq S_1\times S_2$ and $R_2 \subseteq S_2
  \times S_3$ is a bisimulation on $S_1$ and $S_3$.
\end{lemma}
\begin{proof}
  This lemma follows immediately from the following two facts:
  1)~comosition of bisimulations on $F$-coalgebras is a bisimulation
  when $F$ preserves weak pullbacks \cite{Rutten00}; 2)~$\ndf$
  functors preserve weak pullbacks \cite{Sokolova-PhD}.
\end{proof}

Although \thm{bisim:union} allows us to speak of the maximal
bisimulation on two coalgebras, the witness mapping on this maximal
bisimulation need not be unique.  The following proposition provides a
construction of a canonical witness mapping on a bisimulation of two
coalgebras.

\begin{proposition}
  Consider a bisimulation $R \subseteq S_1\times S_2$ on two
  $F$-coalgebras $(S_1,f_1)$ and $(S_2,f_2)$ with $F \in \ndf$.  Let
  $g_1,g_2:R \rightarrow F(R)$ be two witness mappings.  The mapping
  $g_1 \maxwit{F} g_2: R\rightarrow F(R)$ defined below is again a
  witness for bisimulation $R$.
  \begin{eqnarray}
    \label{eq:max:witness}
    g_1 \maxwit{F} g_2 & \bydef{=} &
    \begin{cases}
      g_1 = g_2, 
      & \text{if }F = \id \text{ or } F = \bB\\
      (g_1^1 \maxwit{F_1} g_2^1)\times (g_1^2 \maxwit{F_2} g_2^2)
      & \text{if }F = F_1 \times F_2,\quad g_i = g_i^1 \times g_i^2\ (i =1,2)\\
      \lambda a.\Big(g_1(a) \maxwit{G} g_2(a)\Big), 
      & \text{if }F = G^\actions\\
      g_1 \cup g_2,
      & \text{if }F = \powerset(G),
    \end{cases}
  \end{eqnarray}
  with $g_1 \cup g_2: (s_1,s_2) \mapsto g_1(s_1,s_2) \cup g_2(s_1,s_2)$.
\end{proposition}
\begin{proof}
  To prove the proposition, we have to show that $\maxwit{F}$ is well
  defined, \ie that $g_1 = g_2$ for either $F=\id$ or $F=\bB$, and
  that projections $\pi_i: R \rightarrow S_i$ ($i=1,2$) are coalgebra
  homomorphisms from $(R,g_1 \maxwit{F} g_2)$ to $(S_i, f_i)$.

  For $F = \id$, we have, $g_1 = (\pi_1 \circ g_1, \pi_2 \circ g_1) =
  (f_1 \circ \pi_1, f_2 \circ \pi_2) = (\pi_1 \circ g_2, \pi_2 \circ
  g_2) = g_2$.
  
  For $F = \bB$, we have $F(\pi_1) = id_{\bB}$ and $g_1 = id_{\bB}
  \circ g_1 = f_1 \circ \pi_1 = id_{\bB} \circ g_2 = g_2$.

  The proof that $\pi_i$ are coalgebra homomorphisms is by structural
  induction on $F$.  It is trivial for the cases $F = \id$ and $F =
  \bB$.  For the cases $F = F_1 \times F_2$ and $F = G^\actions$, it
  follows immediately from \rem{subbisimulation}.  Finally, for $F =
  \powerset(G)$, we have $\powerset(G)(\pi_i) \circ (g_1 \cup g_2) =
  (\powerset(G)(\pi_i) \circ g_1) \cup (\powerset(G)(\pi_i) \circ g_2)
  = (f_i \circ \pi_i) \cup (f_i \circ \pi_i) = (f_i \circ \pi_i)$.
\end{proof}

We define the semantic preorder on $F$-coalgebras with $F \in \ndf$ by
putting $(S_1,f_1) \lessint (S_2,f_2)$ iff they have a bisimulation
total on $S_1$.  The meet operator is defined by putting $(S_1,f_1)
\meet (S_2,f_2) \bydef{=} (R, g)$, where $R$ is their maximal
bisimulation and $g: R \rightarrow F(R)$ is the witness mapping
maximal with respect to $\maxwit{F}$ defined by \eq{max:witness}.
Observe that $\maxwit{F}$ can be applied point-wise.  Furthermore, for
each $(s_1,s_2)\in S_1 \times S_2$, there is only a finite number of
possible values for the bisimulation witness mapping.  Hence, $(R,g)$
is defined uniquely.

\begin{proposition}
  $(\factor{\mathbf{Set}_{\ndf}}, \lessint, \meet)$ is a
  meet-semilattice, where $\mathbf{Set}_{\ndf}$ is the category of
  $F$-coalgebra with $F\in \ndf$.
\end{proposition}
\begin{proof}
  First of all, \lem{bisim:comp} implies the transitivity of
  $\lessint$.  To complete the proof, we must show that, in the above
  context, $(R,g)$ is, indeed, the meet of $(S_1,f_1)$ and
  $(S_2,f_2)$.

  Let $(S,f)$ be an $F$-coalgebra such that $(S,f) \lessint
  (S_i,f_i)$, for $i=1,2$.  Then there exist two corresponding
  coalgebras $(R_i,g_i)$ such that each $R_i \subseteq S \times S_i$
  is a bisimulation on $(S,f)$ and $(S_i,f_i)$ total on $S$.  Since
  $R_1$ and $R_2$ are total on $S$, and since $R$ is maximal, $R$ is
  total on the image $R_1(S) \subseteq S_1$.

  Since both $id_R$ and the projection $\pi: R \rightarrow S_1$ are
  coalgebra homomorphisms (the latter by definition of bisimulation),
  \lem{bisim:img} implies that the image $(\pi, id_R)(R)$ is a
  bisimulation on $(S_1,f_1)$ and $(R,g)$.  Moreover, it is total on
  $R_1(S)$.  We conclude, by observing that, by \lem{bisim:comp}, $R_1
  \circ (\pi, id_R)$ is a bisimulation on $(S,f)$ and $(R,g)$, total
  on $S$, \ie $(S,f) \lessint (R,g)$.\footnote{
    Observe that the maximality of $g$ is only used to prove the
    uniqueness of $(R,g)$.  Thus, $(S,f) \lessint (R,g')$, for any
    witness mapping $g': R \rightarrow F(R)$.
}
\end{proof}

In the general case, the subject of coalgebra composition has been
considered, for example, in \cite{Barbosa00, Sokolova08-microcosm}.
Here we provide, for a functor $F \in \ndf$, an example of the
$F$-coalgebra composition operator along the lines of
\cite{Sokolova08-microcosm}.  We define the maximal interaction
operator $\parallel$ by putting, for any two $F$-coalgebras $(S_1,
f_1)$ and $(S_2, f_2)$,
\begin{eqnarray}
  \label{eq:coalgebra:parallel}
  (S_1, f_1) \parallel (S_2,f_2) \bydef{=} (S_1 \times S_2,f_1 \parallel f_2)
  \,,
\end{eqnarray}
with
\begin{equation}
  \label{eq:coalgebra:sync}
  \begin{split}
    (f_1 \parallel f_2)&: S_1 \times S_2 \rightarrow F(S_1 \times S_2)\,\\
    (s_1, s_2) &\mapsto f_1(s_1) \sync{F}{S_1,S_2} f_2(s_2)
  \end{split}
  \ ,
\end{equation}
where $\sync{F}{X,Y} : F(X) \times F(Y) \rightarrow F(X \times Y)$ is
defined by structural induction on $F$.  For any $x \in F(X)$ and $y \in
F(Y)$, we put

\begin{align}
  \label{eq:sync}
  x \sync{F}{X,Y} y & \bydef{=}
  \begin{cases}
    (x,y), & \text{if $F = \id$,}
    \\[5pt]
    x \lor y, & \text{if $F = \bB$,}
    \\[5pt]
    \left(x_1 \sync{F_1}{X,Y} y_1,\ x_2 \sync{F_2}{X,Y} y_2\right), &
    \text{if  
      $F = F_1 \times F_2$, $x = (x_1,x_2)$ and $y = (y_1,y_2)$,}
    \\[5pt]
    \lambda a. \left(x(a) \sync{G}{X,Y} y(a)\right), & 
    \text{if $F = G^\actions$.}
    \\[5pt]
    \left\{
    \left. x' \sync{G}{X,Y} y' \,\right|\, x' \in x, y' \in y
    \right\}, &
    \text{if $F = \powerset(G)$.}
  \end{cases}
\end{align}

As it has been observed in \cite{Sokolova08-microcosm}, the composition
operator defined by \eq{coalgebra:parallel} and \eq{coalgebra:sync} is
well-behaved, provided that $\sync{F}{}$ is a natural transformation of
functors from $F \times F$ to $F$.  In particular, this would guarantee
that this composition operator preserves coalgebra homomorphisms and,
consequently the semantic preorder defined above.

\begin{proposition}
  $\sync{F}{}$ is a natural transformation from $F \times F :
  \mathbf{Set}^2 \rightarrow \mathbf{Set}$ to $F: \mathbf{Set}^2
  \rightarrow \mathbf{Set}$.
\end{proposition}
\begin{proof}
  We prove, by structural induction on $F$ that, for any morphism $h = (h_1
  \times h_2) : X \times Y \rightarrow X' \times Y'$, the diagram
  \[
    \begin{CD}
      F(X) \times F(Y) @>{\sync{F}{X,Y}}>> F(X \times Y)\\
      @V{(F\times F)(h)}VV                 @VV{F(h)}V\\
      F(X') \times F(Y') @>>{\sync{F}{X',Y'}}> F(X' \times Y')
    \end{CD}
  \]
  is commutative, \ie that, for any $x \in F(X)$ and $y \in F(Y)$, holds
  \begin{eqnarray}
    F(h_1 \times h_2)\left(x \sync{F}{X,Y} y\right)
    & = & 
    F(h_1)(x) \sync{F}{X',Y'} F(h_2)(y)\,.
  \end{eqnarray}

  {\bf Case 1.} ($F = \id$)
  \[
  \id(h_1 \times h_2)\left(x \sync{\id}{X,Y} y\right) = 
  (h_1 \times h_2)(x,y) = 
  (h_1(x),h_2(y)) = 
  \id(h_1)(x) \sync{\id}{X',Y'} \id(h_2)(y)\,.
  \]

  {\bf Case 2.} ($F = \bB$)
  \[
  \bB(h_1 \times h_2)\left(x \sync{\bB}{X,Y} y\right) = 
  id_{\bB}(x \lor y) = 
  x \lor y = 
  id_{\bB}(x) \lor id_{\bB}(y) = 
  \bB(h_1)(x) \sync{\bB}{X',Y'} \bB(h_2)(y)\,.
  \]

  {\bf Case 3.} ($F = F_1 \times F_2$) Let $x = (x_1, x_2)$ and $y =
  (y_1,y_2)$.
  \begin{align*}
    (F_1 &\times F_2)(h_1 \times h_2)
    \left((x_1,x_2) \sync{F_1 \times F_2}{X,Y} (y_1,y_2)\right) =\\
    &= \left(F_1(h_1 \times h_2) \times F_2(h_1 \times h_2)\right)
    \left(x_1 \sync{F_1}{X,Y} y_1,\ x_2 \sync{F_2}{X,Y} y_2\right)
    &&\text{by \eq{sync} and the definition in \tab{ndf}}\\
    &= \left(F_1(h_1 \times h_2)\left(x_1 \sync{F_1}{X,Y} y_1\right),\
    F_2(h_1 \times h_2)\left(x_2 \sync{F_2}{X,Y} y_2\right)\right)
    &&\text{by the definition in \tab{operations}}\\
    &= \left(F_1(h_1)(x_1) \sync{F_1}{X',Y'} F_1(h_2)(y_1),\
    F_2(h_1)(x_2) \sync{F_2}{X',Y'} F_2(h_2)(y_2)\right)
    &&\text{by the induction hypothesis}\\
    &= \Big(F_1(h_1)(x_1), F_2(h_1)(x_2)\Big) \sync{F_1 \times F_2}{X',Y'}
    \Big(F_1(h_2)(y_1), F_2(h_2)(y_2)\Big)
    &&\text{by \eq{sync}}\\
    &= (F_1 \times F_2)(h_1)(x_1,x_2) \sync{F_1 \times F_2}{X',Y'}
    (F_1 \times F_2)(h_2)(y_1,y_2)
    &&\text{by the definition in \tab{operations}.}
  \end{align*}

  {\bf Case 4.} ($F = G^\actions$)
  \begin{align*}
    G^\actions(h_1 &\times h_2)\left(x \sync{G^\actions}{X,Y} y\right) = \\
    &= G(h_1 \times h_2) \circ 
    \lambda a. \left(x(a) \sync{G}{X,Y} y(a)\right)
    &&\text{by \eq{sync} and the definitions in 
      Tables~\ref{tab:operations} and \ref{tab:ndf}}\\
    &= \lambda a. \left(G(h_1 \times h_2) 
    \left(x(a) \sync{G}{X,Y} y(a)\right)\right)
    &&\text{by the definition of function composition}\\
    &= \lambda a. \left(G(h_1)(x(a)) \sync{G}{X',Y'} G(h_2)(y(a))\right)
    &&\text{by the induction hypothesis}\\
    &= \Big(\lambda a. G(h_1)(x(a))\Big) \sync{G^\actions}{X',Y'} 
    \Big(\lambda a. G(h_2)(y(a))\Big)
    &&\text{by \eq{sync}}\\
    &= G^\actions(h_1)(x) \sync{G^\actions}{X',Y'} G^\actions(h_2)(y)
    &&\text{by the definitions in Tables~\ref{tab:operations} and 
      \ref{tab:ndf}.}
  \end{align*}

  {\bf Case 5.} ($F = \powerset(G)$)
  \begin{align*}
    \powerset(&G)(h_1 \times h_2)
    \left(x \sync{\powerset(G)}{X,Y} y\right) = \\
    &= \left\{ G(h_1 \times h_2)\left(x' \sync{G}{X,Y} y'\right)
    \left|\begin{array}{l} x' \in x\\y' \in y \end{array}\right.\right\}
    &&\text{by \eq{sync} and the definitions in 
      Tables~\ref{tab:operations} and \ref{tab:ndf}}\\
    &= \left\{ G(h_1)(x') \sync{G}{X',Y'} G(h_2)(y')
    \left|\begin{array}{l} x' \in x\\y' \in y \end{array}\right.\right\}
    &&\text{by induction hypothesis}\\
    &= \left\{ x'' \sync{G}{X',Y'} y'' \left|
    \begin{array}{l} 
      x'' \in \{G(h_1)(x')\,|\,x'\in x\}\\
      y'' \in \{G(h_2)(y')\,|\,y'\in y\}
    \end{array}
    \right.\right\}
    \\
    &= \left\{x'' \sync{G}{X',Y'} y'' \left|
    \begin{array}{l}
      x'' \in \powerset(G(h_1))(x)\\
      y'' \in \powerset(G(h_2))(y) 
    \end{array}
    \right.\right\}
    &&\text{by the definition in \tab{operations}}\\
    & = \powerset(G)(h_1)(x) \sync{\powerset(G)}{X',Y'} \powerset(G)(h_2)(y)
    &&\text{by \eq{sync} and the definition in \tab{ndf}.}  
  \end{align*}
\end{proof}

Finally, we put $\zerobeh_F = (1, f^0_F)$, with $f^0_F: 1 \rightarrow F(1)$
defined, once again, by structural induction on $F$:
\begin{eqnarray*}
  f^0_F(*) & \bydef{=} &
  \begin{cases}
    *, & \text{if $F = \id$,}\\
    \bot, & \text{if $F = \bB$,}\\
    (f^0_{F_1}(*), f^0_{F_2}(*)), & \text{if $F = F_1 \times F_2$,}\\
    \lambda a. f^0_G(*), & \text{if $F = G^\actions$,}\\
    \{f^0_G(*)\}, & \text{if $F = \powerset(G)$,}
  \end{cases}
\end{eqnarray*}
where $\bot \in \bB$ is the bottom of $\bB$.

Clearly, taken together, the elements defined in this section form a
behaviour type over the family of $F$-coalgebras.


\section{Behaviour composition}
\label{sec:composition}

\subsection{Composition operators}
\label{sec:operators}

Assume that a behaviour type $(\cB, \parallel, \issimulatedby, \lessint,
\meet, \zerobeh)$ is given.

\begin{definition}[Composition operator]
  \label{defn:comp}
  An $n$-ary operator $f : \cB^n \rightarrow \cB$ is a \emph{composition
    operator} iff it satisfies the following properties, for any
  $B_1,\dots,B_n,\tilde{B} \in \cB$:
  \begin{enumerate}
  \item \label{comp:simulated} $f(B_1,\dots,B_n) \issimulatedby B_1
    \parallel \dots \parallel B_n$,
  \item \label{comp:preserves} For any $i \in [1,n]$, $B_i \lessint
    \tilde{B}$ implies $f(B_1, \dots, B_i, \dots, B_n) \lessint f(B_1,
    \dots, \tilde{B}, \dots, B_n)$.
  \end{enumerate}
  
  We denote by $\comp$ the set of all composition operators.  $\comp =
  \bigcup_{n\geq 1} \comp[(n)]$, where $\comp[(n)]$ is the set of all $n$-ary
  composition operators.
\end{definition}

Among the immediate consequences of the above definition, one should notice
the following facts.

\begin{lemma}
  \begin{enumerate}
  \item The equivalence relation $\simeq\ =\ \lessint \cap \lessint^{-1}$
    is a congruence with respect to composition operators;
  \item The maximal interaction operator $\parallel$ is a composition
    operator;
  \item For any $B_1,B_2 \in \cB$, one has $B_1 \issimulatedby B_1
    \parallel B_2$.
  \end{enumerate}
\end{lemma}
\begin{proof}
  The first two statements of the lemma are trivial.  The third one is
  proven by observing that $\zerobeh \issimulatedby B_2$ and,
  consequently, $B_1 \simeq B_1 \parallel \zerobeh \issimulatedby B_1
  \parallel B_2$.
\end{proof}

\begin{definition}[Composition of operators]
  \label{defn:comp:compose}
  For an $n$-ary operator $f_1 : \cB^n \rightarrow \cB$, an $m$-ary
  operator $f_2 : \cB^m \rightarrow \cB$, and $i \in [1,n]$, the
  $(n+m-1)$-ary operator $f_1 \circ_i f_2$ is defined by
  \begin{eqnarray}
    \label{eq:comp:compose}
    (f_1 \circ_i f_2)\Big(B_1,\dots,B_{n+m-1}\Big) & \bydef{=} &
    f_1\Big(
      B_1,\dots,B_{i-1}, 
      f_2(B_n,\dots,B_{n+m-1}), 
      B_i,\dots,B_{n-1}
    \Big).\quad
  \end{eqnarray}
\end{definition}

\begin{lemma}
  A composition of two composition operators is also a composition
  operator.
\end{lemma}

\begin{example}[BIP interaction model]
  \label{ex:sos}
  A convenient way of defining composition operators is through the
  use of SOS rules.  For example, consider the behaviours specified by
  LTS (see \secn{lts}).  In the BIP interaction model
  \cite{BliSif07-acp-emsoft}, given $\gamma \subseteq 2^\actions$, the
  corresponding $n$-ary composition operator is defined on behaviours
  $B_i = (Q_i, A_i, \goesto_i)$, for $i\in [1,n]$, by putting
  $\gamma(B_1,\dots,B_n) \bydef{=} (\prod_{i=1}^n Q_i, \bigcup_{i=1}^n
  A_i, \goesto)$, where $\goesto$ is the minimal transition relation
  satisfying the following set of SOS rules
  \begin{equation}
    \label{eq:positive}
    \left\{\left.
    \renewcommand{\arraystretch}{1.5}
    \derrule[3]{
      \Big\{q_i \goesto[a_i]_i q_i'\Big\}_{i \in I} &
      \Big\{q_i = q_i'\Big\}_{i \not\in I} &
      \bigcup_{i \in I} a_i = a
    }{
      q_1 \dots q_n \longgoesto[a] q_1' \dots q_n'
    }
    \ \right|\,
    a \in \gamma
    \right\}\,.
  \end{equation}
\end{example}

\begin{proposition}
  Any operator $f$ defined as above is a composition operator on LTS.
\end{proposition}
\begin{proof}
We have to show that $f$ preserves the preorder $\lessint$.  By
symmetry, it is sufficient to prove that, for $B_1 \lessint
\tilde{B_1}$ and $B_2,\dots,B_n$, we have $f(B_1, B_2,\dots,B_n)
\lessint f(\tilde{B_1}, B_2, \dots, B_n)$.  First of all, since $A_1
\subseteq \tilde{A_1}$, we have $\bigcup_{i=1}^n A_i \subseteq
\tilde{A_1} \cup \bigcup_{i=2}^n A_i$.

Since $B_1 \lessint \tilde{B_1}$, there exists a relation $\cR_1
\subseteq Q_1 \times \tilde{Q_1}$ total on $Q_1$ and satisfying
\eq{lts:less}.  We can then define a relation $\cR \subseteq (Q_1
\times \prod_{i=2}^n Q_i) \times (\tilde{Q_1} \times \prod_{i=2}^n
Q_i)$, by putting $(q_1, q_2, \dots, q_n) \cR (\tilde{q_1}, q_2',
\dots, q_n')$ iff $q_1 \cR_1 \tilde{q_1}$ and $q_i = q_i'$, for
$i\in[2,n]$. This relation is clearly total.  Let $q_1 q_2 \dots q_n
\goesto[a] q_1'\dots q_n'$ be a transition in $f(B_1,B_2,\dots,B_n)$
and let $q_1\cR_1 \tilde{q_1}$.  By definition of the operator $f$,
this transition must be inferred by the rule \eq{positive} from a set
of transitions $\{q_i \goesto[a_i]_i q_i'\}_{i \in I}$ with $a =
\bigcup_{i\in I} a_i$.  If $1 \not\in I$, then clearly $q_1 = q_1'$
and the corresponding transition $\tilde{q_1} q_2 \dots q_n \goesto[a]
\tilde{q_1} q_2'\dots q_n'$ is possible in $f(\tilde{B_1}, B_2, \dots,
B_n)$.  If $1 \in I$, by \eq{lts:less}, there exists a transition
$\tilde{q_1} \goesto[b_1] \tilde{q_1}'$ such that $a_1 \subseteq b_1$
and $q_1' \cR_1 \tilde{q_1}'$. Hence, $a = a_1 \cup \bigcup_{i \in
  I\setminus\{1\}}a_i \subseteq b \cup \bigcup_{i \in
  I\setminus\{1\}}a_i \bydef{=} b$ and, by \eq{positive}, $\tilde{q_1}
q_2 \dots q_n \goesto[b] \tilde{q_1}' q_2\dots q_n$, which proves the
proposition.
\end{proof}

\begin{example}[Negative premises]
  Consider the family of SOS operators with negative premises, that is
  defined by the rules of the form
  \begin{equation}
    \label{eq:negative}
    \renewcommand{\arraystretch}{1.5}
    \derrule[4]{
      \Big\{q_i \goesto[a_i]_i q_i' \,\Big|\, i \in I\Big\} &
      \Big\{q_i = q_i' \,\Big|\, i \not\in I\Big\} &
      \Big\{q_j \not\longgoesto[b^k_j]_j \,\Big|\, j \in J, k \in K_j\Big\} &
      a = \bigcup_{i\in I} a_i
    }{q_1 \dots q_n \longgoesto[a] q_1' \dots q_n'}\,,
  \end{equation}
  where $q_j \not\longgoesto[b^k_j]_j$ means that there is no
  transition labelled $b^k_j$ possible from the state $q_j$ of
  behaviour $B_j$.

  It is easy to see that such operators do not preserve simulation
  relation $\lessint$ and therefore are not composition operators on
  the LTS behaviour type as defined in \secn{lts}.  An adaptation of
  the ready simulation relation from the initial proposition by Bloom
  \cite{bloom-phd} is necessary to obtain a behaviour type, for which
  such operators would, indeed, be composition operators.
\end{example}

Let us now revisit the {\em separation of concerns} principle
mentioned in the introduction.  In our context, this principle
consists in separating the computation of a system from the
application of {\em glue operators} coordinating its atomic
components.  Intuitively, this means that no additional behaviour
should be introduced by application of a glue operator.  Not only glue
operators are limited to restricting the behaviour of atomic
components by imposing coordination constraints (\cf condition
\ref{comp:simulated} of \defn{comp}), but, on top of that, they should
be {\em memoryless}.  Intuitively, this corresponds to requiring that
two conditions be satisfied: 1)~the glue operator should not add {\em
  state} to the coordinated system and 2)~the actions possible in a
global state of the composed system are completely determined by the
properties of the corresponding states of the constituent subsystems.
The latter condition is justified by the fact that a memoryless
operator is unaware of any states of the system other than the current
one.  In order to impose these additional requirements, we need a
formal notion of state, provided precisely by coalgebras.

\begin{definition}
  \label{defn:glop}
  Let $F$ be a $\mathbf{Set}$-functor.  An $n$-ary composition
  operator $gl$ on a behaviour type over the family of $F$-coalgebras
  is a {\em glue operator} iff there exists a natural transformation
  $sync: F^n \rightarrow F$, such that, for any set $\{B_i = (S_i,
  f_i)\,|\,i \in [1,n]\}$ of $F$-coalgebras, $gl(B_1,\dots,B_n) =
  (S, f)$ with $S = \prod_{i=1}^n S_i$ and $f(s) =
  sync(f_1(s_1),\dots,f_n(s_n))$, for all $s = (s_1,\dots,s_n) \in S$.
\end{definition}

We now consider two classical operators, namely the {\em prefixing
  operator} and {\em choice}.  The former is a unary operator, which
consists in executing a given action before ``running'' the behaviour
to which the prefixing is applied.  The latter is an associative and
commutative binary operator, which consists in running exactly one of
the behaviours to which it is applied.  In order to implement either
of these two operators, one has to "remember" that, respectively, the
initial transition or the choice of the two components has been made,
thus adding state to the composed system.

For behaviour types based on coalgebras, the binary choice operator
$+$ can be formally defined as the coproduct functor on the category
of coalgebras, \ie $(S_1, f_1) + (S_2, f_2) \bydef{=} (S_1 + S_2, f_1
+ f_2)$.

\begin{hypothesis}
  Let $F$ be a $\mathbf{Set}$-functor weakly preserving pullbacks.  In
  a behaviour type over $F$-coalgebras with the semantic preorder
  defined as in \secn{coalgebras} there is no glue operator $gl$, in
  the sense of \defn{glop}, such that $gl \simeq +$.
\end{hypothesis}

To formulate a similar hypothesis for the prefixing operator, a notion
of {\em initial state} is necessary.  Although such a notion is
provided by {\em pointed coalgebras}, we omit this discussion in this
paper.  

Notice also the distinction between choice and {\em interleaving}.
The former consists in choosing once and for all the behaviour to run,
whereas the latter makes this choice independently at each execution
step.

\subsection{Combining composition operators}
\label{sec:combining}

Although composition of operators introduced in \defn{comp:compose}
allows to combine several operators hierarchically, it does not allow
``simultaneous'' application of the constraints imposed by two
operators.

\begin{example}[Simultaneous application of two operators]
  \label{ex:simultaneous}
  Indeed, consider, for example, for $i \in [1,4]$, $B_i = (Q_i, A_i,
  \goesto) \in LTS$ (see \secn{lts}) and four given actions $a_i$,
  such that $a_i \in A_i$ and $a_i \not\in A_j$ for $i \neq j$.
  Consider also two quaternary composition operators $f_1, f_2 \in
  \comp[(4)]$, such that the action of $f_1$ consists in synchronising
  the actions $a_1$ and $a_2$ (\cf \ex{sos}) of the first two
  components it is applied to, whereas $f_2$ synchronises the actions
  $a_3$ and $a_4$ of the last two of its arguments.  More precisely,
  $f_1(B_1,B_2,B_3,B_4) = \gamma_{a_1 a_2} (B_1,B_2) \parallel B_3
  \parallel B_4$ and $f_2(B_1,B_2,B_3,B_4) = B_1 \parallel B_2
  \parallel \gamma_{a_3 a_4} (B_3,B_4)$, where the binary operator
  $\gamma_a$, parametrised by an interaction $a$ (here, $a = a_1 a_2$,
  for $f_1$, and $a = a_3 a_4$, for $f_2$) is defined by the following
  rules, $x$ and $y$ being action variables,
  \begin{gather}
    \begin{array}{*{3}{c}}
      \derrule[2]
              {q_1 \goesto[x] q_1' & x \not\in a}
              {q_1q_2 \longgoesto[x] q_1' q_2}\,, &
      \derrule[2]
              {q_2 \goesto[x] q_2' & x \not\in a}
              {q_1q_2 \longgoesto[x] q_1 q_2'}\,, &
      \derrule[3]
              {q_1 \goesto[x] q_1' & q_2 \goesto[y] q_2' & 
                x,y \not\in a}
              {q_1q_2 \longgoesto[x \cup y] q_1' q_2'}\,,
    \end{array}
    \\
    \derrule[2]
            {\Big\{q_i \longgoesto[a \cap A_i] q_i' \,\Big|\, a\cap A_i \neq \emptyset\Big\} &
              \Big\{q_i = q_i' \,\Big|\, a\cap A_i = \emptyset\Big\}}
            {q_1q_2 \longgoesto[a] q_1' q_2'}\,.
  \end{gather}

  Intuitively, simultaneous application of $f_1$ and $f_2$ to any
  components $B_1$, $B_2$, $B_3$, $B_4$ should enforce, on one hand,
  the synchronisation of $a_1$ and $a_2$ in $B_1$ and $B_2$
  respectively, and, on the other hand, the synchronisation of
  respectively $a_3$ and $a_4$ in $B_3$ and $B_4$.  However, the
  result $f_1(B_1,B_2,B_3,B_4)$ of applying $f_1$ is a single
  component, to which $f_2$ cannot be applied any more, since the
  latter is a quaternary operator.

  Furthermore, considering the specification of the $\gamma_a$
  operator above as {\em ``enforcing the synchronisation of actions
    belonging to the interaction $a$''}, one would expect the
  simultaneous application of $\gamma_{a_1 a_2}$ and $\gamma_{a_3
    a_4}$ to any set of behaviours to achieve the same effect as
  above, without having to explicitly define operators $f_1$ and $f_2$
  and avoiding the associated problem discussed in the previous
  paragraph.
\end{example}

First of all, the semantic preorder $\lessint$ and the operator
$\meet$ can be canonically extended to composition operators, provided
they have the same arity: for any $n \geq 1$ and $gl_1, gl_2 \in
\comp[(n)]$,
\begin{gather}
  \label{eq:comp:refine}
  f_1 \lessint f_2
  \ \bydef{\Longleftrightarrow} \
  \forall B_1,\dots,B_n \in \cB,\
  \Big(f_1(B_1,\dots,B_n) \lessint f_2(B_1,\dots,B_n)\Big)\,,
  \\
  \label{eq:comp:meet}
  \forall B_1,\dots,B_n \in \cB,\ (f_1 \meet f_2)(B_1,\dots,B_n)
  \ \bydef{=} \
  f_1(B_1,\dots,B_n) \meet\, f_2(B_1,\dots,B_n)\,.
\end{gather}

\begin{proposition}
  $(\factor{\comp[(n)]}, \lessint, \meet)$ is a meet-semilattice.
\end{proposition}
\begin{proof}
  Let $f_1,f_2 \in \comp[(n)]$ be two composition operators and
  consider $f \in \comp[(n)]$ such that $f_1\meet f_2 \lessint f
  \lessint f_1, f_2$.  The right-hand relation implies that, for any
  $B \in \cB$, $f(B) \lessint f_1(B)\meet f_2(B) = (f_1 \meet
  f_2)(B)$.  Hence $f \lessint f_1 \meet f_2$ and, together with the
  left-hand relation, this implies $f \simeq f_1 \meet f_2$.
\end{proof}

\begin{example}[Preorder on SOS operators]
  \label{ex:sos:preorder}
  Consider once again the $LTS$ behaviour type defined in \secn{lts}
  and the corresponding family of composition operators defined by
  \eq{positive} in \ex{sos}.  Identifying each operator with the set
  of its defining SOS rules, it follows directly from \cite[Lemma
    3]{BliSif08-express-concur} that, for two such operators $f_1$ and
  $f_2$, $f_1 \lessint f_2$ is equivalent to $f_1 \subseteq f_2$.
  Hence, $f_1 \meet f_2 = f_1 \cap f_2$.
\end{example}

Going back to \ex{simultaneous}, $(f_1 \meet f_2)(B_1,B_2,B_3,B_4)$
represents the behaviour where both the actions of $B_1$ are synchronised
with those of $B_2$, and the actions of $B_3$ are synchronised with those
of $B_4$, as defined respectively by $f_1$ and $f_2$.

In order to allow application of a given composition operator to any set of
component behaviours with cardinality at least the arity of the composition
operator in question, we introduce below the {\em arity extension} for
composition operators.

\begin{definition}
  The \emph{arity extension} of an $n$-ary composition operator $f \in
  \comp[(n)]$ to arity $m \geq n$ is the composition operator $\arity{m}{f}
  \in \comp[(m)]$ defined by putting, for all $B_1,\dots,B_n \in \cB$,
  \begin{eqnarray}
    \label{eq:extension}
    \arity{m}{f}(B_1,\dots, B_n)
    & \bydef{=} 
    & \bigmeet_{\sigma \in S_m} \Big(
    f\left(B_{\sigma(1)},\dots,B_{\sigma(n)}\right) \parallel 
    B_{\sigma(n+1)} \parallel \dots \parallel B_{\sigma(m)}\Big)\,,
  \end{eqnarray}
  where $S_m$ is the group of all permutations of $[1,m]$.  The right-hand
  side of \eq{extension} consists in simultaneously applying $f$ to all
  possible subsets of $n$ components.
\end{definition}

For operators $f_1$ and $f_2$ of \ex{simultaneous}, we have $f_1 =
\arity{4}{(\gamma_{a_1 a_2})}$ and $f_2 = \arity{4}{(\gamma_{a_3 a_4})}$.

\begin{lemma}[Isotony of arity extension]
  Arity extension preserves the semantic preorder, that is, for any
  $f_1,f_2 \in \glue[(n)]$ and $m \geq n$, $f_1 \lessint f_2$ implies
  $\arity{m}{f_1} \lessint \arity{m}{f_2}$.
\end{lemma}

\begin{example}[Disjoint behaviours]
  \label{ex:disjoint}
  Going back to \ex{simultaneous} of \secn{combining} and considering that,
  for $i \neq j$, $a_i \not\in A_j$, one can see that
  \[
  f_1 \meet f_2 \simeq 
  \arity{4}{\gamma_{a_1 a_2}} \meet \arity{4}{\gamma_{a_3 a_4}} \simeq
  \arity{4}{(\gamma_{a_1 a_2} \meet \gamma_{a_3 a_4})}\,.
  \]
  The assumption that, for $i \neq j$, $a_i \not\in A_j$ is essential here.
  It guarantees that the operators $\gamma_{a_1 a_2}$ and $\gamma_{a_3
    a_4}$ do not {\em interfere} with the behaviour of components $B_3,
  B_4$ and $B_1, B_2$ respectively.  Precise characterisation of such a
  notion of {\em non-interference} will be an important part of our future
  work.
\end{example}

\subsection{Symmetrical composition operators}
\label{sec:symmetrical}

Notice that the operator $\arity{m}{f}$ defined by \eq{extension} is
symmetrical in the sense of the following definition.

\begin{definition}
  An operator $f: X^n \rightarrow Y$ is called {\em symmetrical} iff, for
  any permutation $\sigma \in S_n$ and any $x_1,\dots,x_n \in X$, holds the
  equation $f(x_1,\dots,x_n) = f(x_{\sigma(1)},\dots, x_{\sigma(n)})$.
\end{definition}

Clearly, an operator $f \in \comp[(n)]$ is symmetrical iff $f \simeq
\arity{n}{f}$.  Furthermore, for symmetrical operators, one can also
unambiguously define the arity reduction.

\begin{definition}
  The \emph{arity reduction} of a symmetrical $n$-ary composition operator
  $f \in \comp[(n)]$ to arity $m \leq n$ is the composition operator
  $\arity{m}{f} \in \comp[(m)]$ defined by putting, for all $B_1,\dots,B_m
  \in \cB$,
  \begin{eqnarray}
    \label{eq:reduction}
    \arity{m}{f}(B_1,\dots, B_n)
    & \bydef{=} &
    f(B_1,\dots, B_n, \zerobeh,\dots, \zerobeh)\,.
  \end{eqnarray}
\end{definition}

Given a symmetrical operator $f \in \comp$, one can combine both
concepts\mdash arity extension and reduction\mdash to define the operator
$\tilde{f}: \powerset(\cB) \rightarrow \cB$, by putting, for any finite
subset $\bB \subseteq \cB$, $\tilde{f}(\bB) \bydef{=}
\arity{|\bB|}{f}(\bB)$.

Observe that several popular parallel composition operators, such as the
ones used in CCS \cite{milner89}, CSP \cite{hoare85} and BIP
\cite{BliSif07-acp-emsoft}, are symmetrical.


\section{Discussion and related work} 
\label{sec:discussion}

The complete bibliography, as related to the motivations behind the
present paper, is yet to be established.  However, several
contributions can already be mentioned.

Based on the same observations about the importance of studying glue
as a first-class notion, Abstract Behavior Types (ABM) were proposed
in \cite{arbab05-abt} guiding the design of the Reo language.  ABT
propose to characterise components as channels (or dataflow
transformers) that do not provide any information as to the manner in
which the defining transformations are computed.  This approach is
radically different from the one taken in the present paper, since the
main emphasis is put on one aspect of component behaviour (namely the
dataflow transformation) and its expressive power with regards to
assembling more complex transfer functions.  In particular composition
of channels boils down to pipeline assembly.  Although, ABT are likely
to provide another interesting case study for the behaviour types we
propose, they lack the abstraction necessary to model a larger class
of behaviour as intended here.

It would be interesting to further verify the robustness of the
proposed framework by defining behaviour types based on interface
\cite{deAlfaro01-interface-aut} or modal \cite{larsen07-modal-aut}
automata, whereof the particularity is that the corresponding standard
refinement relations are contravariant as opposed to the examples of
this paper.

Furthermore, although all the examples provided in this paper can be
interpreted as state-based behaviour types and, therefore, modelled as
coalgebras of the appropriate types, the intention is to keep the
abstraction level sufficiently high in order to be able to accommodate
for behaviour types that do not have a clearly identifiable notion of
state.  In particular, it would be interesting to investigate
applicability of this framework to continuous time systems as modelled
in Simulink\footnote{http://www.mathworks.com/products/simulink/} or
Modelica\footnote{http://www.modelica.org/}.

As mentioned above, most\mdash if not all\mdash behaviour types, of
interest for the author of this paper, can be modelled as coalgebras
of a suitable type.  Since their introduction as a model for system
behaviour, coalgebras have been subject to extensive studies.  It
seems, however, that most of these studies were focusing primarily on
coalgebras as a way to model individual component behaviour.
Although, in \cite{Barbosa00, Sokolova08-microcosm}, coalgebra
composition has been addressed in a much more general form than in the
present paper, some questions remain unanswered: {\em What is the
  class of coalgebra types, for which a meaningful maximal interaction
  operator can be constructively defined?}  {\em How does one
  characterise the composition operators other than maximal
  interaction?} The latter question is related to the work presented
in \cite{hasuo11-microcosm,turi97-math-op-sem}, where GSOS-style
operational semantics is studied from the categorical point of view.

Last but not least, an important subject that we have mentioned in
this paper and that we are planning to address as part of our future
work is the interference between glue operators.  Consider two glue
operators $gl_1$ and $gl_2$, and a family of behaviours $\{B_i\}_{i =
  1}^n$.  Assume, furthermore, that, for $i =1,2$,
$\arity{n}{gl_i}(B_1,\dots,B_n)$ satisfies some given property $P_i$.
{\em What conditions have to be satisfied by $gl_1$ and $gl_2$, on one
  hand, and $\{B_i\}_{i = 1}^n$, on the other hand, for any of the
  composite behaviours
  $\arity{1}{gl_i}(\arity{n}{gl_j}(B_1,\dots,B_n))$ ($i \neq j$),
  $\arity{n}{(gl_1 \meet gl_2)}(B_1,\dots,B_n)$, \etc to satisfy $P_1
  \land P_2$?} {\em Given computable representations of $gl_1$ and
  $gl_2$ how does one compute the operator imposing $P_1 \land P_2$?}
The latter question was partially addressed in \ex{sos:preorder} of
\secn{combining}.  Furthermore, our result in
\cite{BliSif08-express-concur} suggests that this example can be
generalised, on a suitable behaviour type, to SOS operators with
negative premises.  {\em Can a similar result be obtained for a larger
  class of coalgebras types?}

A first, na\"{\i}ve approach consists in defining a class of {\em
  distributive} behaviour types characterised by the distributivity of
their maximal interaction operator over the meet operator, that is,
for any $B_1,B_2,B_3 \in \cB$,
  \begin{eqnarray*}
    (B_1 \meet B_2) \parallel B_3 
    & \simeq & 
    (B_1 \parallel B_3)\ \meet\ (B_2 \parallel B_3)\,.
  \end{eqnarray*}
It is easy to see that distributivity of a behaviour type implies that
of the arity extension of composition operators over the meet
operator: for any $f_1,f_2 \in \comp[(n)]$ and $m \geq n$,
\[
\arity{m}{(f_1 \meet f_2)} 
\simeq 
\arity{m}{f_1} \meet \arity{m}{f_2}\,.
\]
Assuming that we know how to compute $f_1 \meet f_2$, we can then
reasonably expect the obtained operator to satisfy both properties
imposed by $f_1$ and $f_2$.

Although none of behaviour types considered in this paper appear to be
distributive, all the counter-examples we have considered while
preparing the paper were based on the fact that several components
shared certain actions.  As illustrated by \ex{disjoint} of
\secn{combining}, ``local distributivity'' can be achieved provided
the ``absence of conflicts'' in the component interfaces.  An
important question to be addressed in the future work is: {\em How can
  these notions of ``local distributivity'' and ``absence of
  conflicts'' be formalised and generalised to larger classes of
  behaviour types?}


\section{Conclusion} 
\label{sec:conclusion}

The goal of this paper was to make a first step towards the definition
of a formal framework for studying behaviour composition in a setting
sufficiently general to provide insight into how the component-based
systems should be modelled and compared.

We have proposed the notions of {\em behaviour type} and {\em
  composition operator}, which, while striving for generality, allow
to capture some essential properties expected when reasoning
intuitively about component composition.  We have illustrated the
notion of behaviour type on three examples, namely Traces, Labelled
Transition Systems and $F$-coalgebras with a restricted class of
non-deterministic functors.

In the framework proposed in this paper, a {\em behaviour type} is a
tuple $(\cB, \parallel, \issimulatedby, \lessint, \meet, \zerobeh)$,
where $\cB$ is the set of underlying behaviours (Traces, Labelled
Transition Systems, Coalgebras, etc.); $\parallel$ is the maximal
interaction operator
defining the joint behaviour of two components without any
coordination constraints; $\issimulatedby$ and $\lessint$ are two
preorders used respectively to represent ``containment'', or
simulation, relation between two behaviours and the semantic relation
reflecting the fact that one component can be substituted by another
one while essentially preserving the intended behaviour.  The relation
$\simeq\ =\ \lessint \cap \lessint^{-1}$ is a congruence for
composition operators.  The example provided in \secn{coalgebras}
supports our idea that these two preorders need not necessarily be the
same.  We require that $(\factor{\cB}, \lessint, \meet)$ be a
meet-semilattice, whereby the meet operator $\meet$ identifies the
common behaviour of two components.  Finally, the notion of zero
behaviour $\zerobeh$ serves as a ``sanity check'' for the behaviour
type.  Intuitively, it represents a component that does nothing.  On
one hand, it should not influence the behaviour of other components
when placed in parallel (for any $B \in \cB$, $B \parallel \zerobeh
\simeq B$) and, on the other hand, it should be simulated by all other
behaviours (for any $B \in \cB$, $0 \issimulatedby B$).

The requirement that the semantic preorder of a given behaviour type
induce a meet-semilattice structure has allowed us to define a meet
of composition operators, representing their simultaneous application
to a given set of behaviours and, consequently, to extend any
composition operator to a symmetrical one, applicable to any finite
set of behaviours.  This, together with the definition of {\em
  memoryless glue operators}, takes us one step closer to a formal
understanding of the separation of concerns principle that we have
advocated in our previous papers, and which stipulates that the
computational aspects of the system should be localised in the atomic
components, whereas the coordination layer responsible for managing
concurrency should be realised by memoryless glue operators.

Finally, we have discussed some related work and key questions,
arising from the proposed framework, that are important for the
understanding of fundamental principles of component-based design.


\section*{Acknowledgements} 
I would like to express my gratitude to the anonymous reviewers for
the instructive discussion on the ICE 2012 forum.  I would also like
to thank Ana Sokolova and Alexandra Silva for the discussion after the
workshop and the pointers to the literature allowing me to include
the powerset functor in \secn{coalgebras}.

\bibliographystyle{eptcs}
\bibliography{glue,express,constraints,connectors,reo,bip,latex8}
\end{document}